\documentclass[aps,pra,superscriptaddress,nofootinbib,showpacs,floatfix,twocolumn,notitlepage,pdftex]{revtex4-2}

\usepackage{graphicx}
\usepackage{dcolumn}
\usepackage{bm}
\usepackage[pdftex, colorlinks=true]{hyperref}

\usepackage[utf8]{inputenc}
\usepackage[english]{babel}
\usepackage[T1]{fontenc}
\usepackage{svg}
\usepackage{amsmath}
\usepackage{amssymb}
\usepackage{color}
\newcommand{\Lind}{\mathbf{\mathcal{L}}}
\newcommand{\E}{\mathbf{\mathcal{E}}}

\usepackage{amsthm}

\newtheorem{prop}{Proposition}
\usepackage{dsfont}

\usepackage{algpseudocode}
\usepackage{mdframed}
\newcounter{algorithmctr}

\begin{document}

\preprint{APS/123-QED}

\title{Trajectory-independent speed limits for controlled open quantum systems}

\author{James B. Larsen}
\email{jblarsen@umich.edu}
\affiliation{Department of Mathematics, University of Michigan, Ann Arbor, Michigan 48109, USA}
\affiliation{Quantum Algorithms and Applications Collaboratory, Sandia National Laboratories, Albuquerque, New Mexico 87185, USA}
\author{Tameem Albash}
\affiliation{Quantum Algorithms and Applications Collaboratory, Sandia National Laboratories, Albuquerque, New Mexico 87185, USA}
\author{Alicia B. Magann}
\affiliation{Quantum Algorithms and Applications Collaboratory, Sandia National Laboratories, Albuquerque, New Mexico 87185, USA}
\author{Christian Arenz}
\affiliation{School of Electrical, Computer, and Energy Engineering, Arizona State University, Tempe, Arizona 85287, USA}

\date{\today}

\begin{abstract}
Existing quantum speed limits for controlled open quantum systems depend on the specified trajectory. For example, lower bounds on quantum annealing times in the presence of dissipation depend explicitly on the chosen annealing schedule. Recently, schedule-independent speed limits have been derived for annealing in the closed quantum system setting (SciPost Phys. 18, 159 (2025)). In this work, we generalize these results to open quantum systems, deriving schedule-independent lower bounds for quantum annealing times in systems described by a Lindblad master equation. We analyze the interplay between coherent control and dissipation in single- and two-qubit examples, demonstrating that the derived lower bounds capture key scaling behavior with respect to the strength of the dissipator. Finally, we apply the bound to thermal state preparation and show that the bound matches the expected asymptotic behavior for an Ising model in the high temperature limit.
\end{abstract}

\maketitle

\section{\label{sec:intro} Introduction}
Time-optimal quantum evolutions have been found by the quantum control community for decades since the succinct quantum brachistochrone problem formulation by Carlini et al. \cite{Carlini2006}, with the generalization to the open quantum system case happening soon after \cite{Carlini2008}.
Related questions about time optimality independently gave rise to the study of quantum speed limits  \cite{Frey2016, Deffner_2017}, seeking to bound the minimum time required to evolve from an initial state to a desired target state. The paradigmatic quantum speed limits are given by the Mandelstam-Tamm \cite{Mandelstam1991} and Margolus-Levitin inequalities \cite{Margolus1998} that lower bound the evolution time $T$ by the inverse of the expectation value and the variance, respectively, of the Hamiltonian governing the evolution. While these bounds and similar results for open quantum systems from the literature \cite{Taddei2013, delCampo2013, Deffner2013, Sun2015, Pires2016PRX, Funo2019, Campaioli2019tightrobust, Mai2023, Rosal2025, Srivastav_2025} can tightly characterize the required evolution time when the dynamics of the system are known \emph{a priori}, practical difficulties arise in the presence of time-dependent control.

One example where such difficulties arise in the presence of time-dependent control is in quantum annealing~\cite{Apolloni1989,Finnila1994,Kadowaki1998,Brooke1999,Farhi2001,Santoro2002,Lucas2014}, a powerful framework for solving optimization problems. 
In this framework, one drives a system from an easy-to-prepare initial state to the solution of the optimization problem encoded in some target state via an annealing schedule. The adiabatic theorem \cite{Born1928,Kato1950,Jansen2007} guarantees that if the initial state is the ground state of some Hamiltonian $H_1$ and the target state is the ground state of a Hamiltonian $H_0$, a sufficiently slow continuous interpolation from $H_1$ to $H_0$ will prepare the desired ground state of $H_{0}$. Quantum annealing through adiabatic schedules is in fact equivalent to universal quantum computing \cite{Aharonov2007, Albash2018}, thus showcasing the power and flexibility of this framework.

When viewed as an algorithm for state preparation, the time complexity of quantum annealing is dictated by the required duration of the schedule. Quantum optimal control can be leveraged to find shorter annealing schedules that violate the adiabatic theorem  \cite{Werschnik2007, GuryOdelin2019}. Evaluating the optimality of a chosen schedule would require a schedule-independent lower bound on the annealing time required, a problem that was addressed recently for the closed quantum system setting in Ref.~\cite{Garc_a_Pintos_2025}. 

Open quantum system dynamics, e.g., dissipation, can often be leveraged to achieve faster state preparation times \cite{Dickson2013, Lin2025}, but the corresponding speed limits for open quantum systems depend on the trajectory taken in Hilbert space \cite{Taddei2013, delCampo2013, Deffner2013, Sun2015, Funo2019, Pires2016PRX, Campaioli2019tightrobust, Mai2023, Srivastav_2025, Rosal2025}. Here we directly extend and generalize the results of Ref.~\cite{Garc_a_Pintos_2025} to quantum annealing in open quantum systems, deriving schedule-independent speed limits in systems described by a Lindblad master equation~\cite{GKS,Lindblad1976} and thus providing necessary conditions to complement the sufficient conditions of the adiabatic theorem. This restriction to the annealing setting presents an avenue to derive trajectory-independent speed limits for non-unitary evolutions.

This work is organized as follows. We start in Sec.~\ref{sec:speedlimitsopensys} by deriving speed limits for controlled open quantum systems. These general considerations are followed by considering the annealing framework in an open quantum system setting, presenting the main result in Eq.~\eqref{eq:specineq}: a schedule-independent lower bound for the quantum annealing time $T$. We go on in Sec. \ref{sec:casestudies} to investigate in numerical simulations the tightness of the derived lower bound for one- and two-qubit systems, focusing on the interplay between coherent control and non-unitary dynamics. Afterwards, we analyze the temperature dependence of the bound in the case of thermal state preparation of an Ising Hamiltonian. Finally, we conclude in Sec.~\ref{sec:conclusion} with the main takeaways and potential future directions of inquiry.

\section{Speed limits for controlled open quantum systems}\label{sec:speedlimitsopensys}

\subsection{Controlled open quantum system dynamics}
We consider a controlled open quantum system \cite{Koch2016} described by the Lindblad master equation \cite{GKS, Lindblad1976, Breuer2007}, 
\begin{equation}
\label{eq:lindblad}
    \frac{d\rho}{dt} = \Lind_{t} \rho,\quad \Lind_{t}\rho = -i[H(t),\rho] + \mathcal{D}(\rho),
\end{equation}
where $H(t)$ is the time-dependent Hamiltonian that describes coherent control and  $\mathcal D$ is the dissipator. This dissipator is assumed to be in Lindblad form and describes non-unitary processes, including dissipation and decoherence.  The solution  of the master equation up to time $t$ is given by a completely-positive and trace-preserving (CPTP) map $\E_t$.  

Looking ahead to the goal of deriving a schedule-independent speed limit, we first derive a bound relative to a secondary system that will eventually help us cancel out the time-dependent dynamics from Eq.~\eqref{eq:lindblad}. As such, we consider another closed quantum system whose unitary dynamics are governed by $\frac{d\rho}{dt} = \tilde{\mathcal L}_{t}\rho=-i[\tilde{H}(t),\rho]$, where $\tilde{H}(t)$ is the Hamiltonian describing the secondary system and $\mathcal{\tilde U}_t$ is the corresponding solution given by a unitary quantum channel. If we assume that the initial state $\rho_{0}$ is a fixed point of $\tilde{\mathcal U}_t$, i.e. $\tilde{\mathcal U}_{t}\rho_{0}=\rho_{0}$, then
\begin{align}
    ||\rho_T - \rho_0||_1&= ||\E_T\rho_0 - \tilde{\mathcal U}_T\rho_0||_1 \nonumber \\
    &= \left|\left| \left( \int_0^T \tilde{\mathcal U}_{-t}(\tilde{\Lind}_{t}-\Lind_{t})\E_{t}\,dt \right) \rho_0 \right|\right|_1 \label{eq:duhamel} \\
    &\leq \int_0^T ||(\tilde{\Lind}_{t}-\Lind_{t})\E_{t} \rho_0 ||_1\,dt. \label{eq:uninvtriineq}
\end{align}
Eq.~\eqref{eq:duhamel} follows from a direct application of the product rule,
\begin{align}
  \frac{d}{dt}(\tilde{\mathcal U}_{-t}\E_t - \mathds{1}) &= \left(\frac{d}{dt}\tilde{\mathcal U}_{-t}\right)\E_t + \tilde{\mathcal U}_{-t} \left(\frac{d}{dt}\E_t\right) \nonumber \\
    &= \tilde{\mathcal U}_{-t} (\tilde\Lind_{t} - \Lind_{t}) \E_t, \label{integraltrick}
\end{align}
where for unitary dynamics $\mathcal U_{t}^{-1}=\mathcal U_{-t}$. This trick is similar to Duhamel's principle, see e.g. Refs.~\cite{burgarth2022one,Wiedmann2026}. Eq.~\eqref{eq:uninvtriineq} follows from unitary invariance of the trace norm $\Vert \cdot \Vert_{1}$ and the triangle inequality. Since
\begin{align}
||(\tilde{\Lind}_{t}-\Lind_{t})\E_{t} \rho_0 ||_1\leq  ||\tilde{\Lind}_{t} - \Lind_{t}||_{1\rightarrow1},
\end{align}
where $\Vert \cdot \Vert_{1\rightarrow 1}$ denotes the induced 1-norm (see App.~\ref{app:norms}), we arrive at 
\begin{align}
 ||\rho_T - \rho_0||_{1} \leq \int_0^T ||\tilde{\Lind}_{t} - \Lind_{t}||_{1\rightarrow1}\,dt \label{eq:diffgenboundint},
\end{align}
which yields the lower bound
\begin{equation}
\label{eq:genineq}
    T \geq \frac{||\rho_T - \rho_0||_1}{\max\limits_{t \in [0,T]} ||\Lind_{t} - \tilde{\Lind}_{t}||_{1\rightarrow1}}.
\end{equation}
We note that this inequality provides a whole family of lower bounds for the time $T$ to prepare a desired target state in an open quantum system described by the master equation in Eq.~\eqref{eq:lindblad}. For each final state $\rho_{T}$, secondary closed system $\tilde{\mathcal L}_{t}$, and initial fixed point $\rho_{0}$ of $\tilde{\mathcal L}_{t}$, Eq.~\eqref{eq:genineq} provides us with a new lower bound. In fact, choosing $\tilde{\mathcal L}_{t}=0$ leads to a speed limit, 
\begin{equation}\label{eq:mtequiv}
    T \geq \frac{||\rho_T - \rho_0||_1}{\frac{1}{T}\int_0^T ||\Lind_{t}||_{1\to 1}\,dt},
\end{equation}
that lower bounds the time $T$ and resembles known speed limits for open quantum systems that depend on the trajectory. For comparison, lower bounds in Refs.~\cite{Deffner2013, Taddei2013, delCampo2013, Pires2016PRX, Campaioli2019tightrobust, Rosal2025} all take the form
\begin{equation}
\label{eq:deffnerboundform}
    T \geq \frac{d\left(\rho_T, \rho_0\right)}{\frac{1}{T}\int_0^T||\Lind_t\rho_t||\,dt},
\end{equation}
where $d$ is some metric quantifying the distance between the initial and final states and the choice of norm in the denominator may vary. For example, Ref.~\cite{Deffner2013} uses the Bures angle $d(\rho_T, \rho_0) = \arccos\left( \sqrt{\langle\psi_0|\rho_T|\psi_0 \rangle}\right)$ between the final state and a pure initial state $\rho_0 = |\psi_0 \rangle \langle \psi_0 |$ for the metric $d$. The denominator term is sometimes referred to as the strength of the generator \cite{Campaioli2019tightrobust} and for unitary processes it equals the time-averaged variance when the Schatten-2 norm is used \cite{Deffner2013}. Furthermore, App.~B of Ref.~\cite{Campaioli2019tightrobust} shows that even speed limits without an explicit inverse dependence on the strength of the generator such as those in Ref.~\cite{Sun2015} can still be upper bounded by this strength. Note that the bound in Eq.~\eqref{eq:mtequiv} maximizes over the strength, loosening the bound and moving toward trajectory independence.

While the generality and loose assumptions of Eq.~\eqref{eq:genineq} are desirable as a mathematical tool, its practical applicability is limited since the denominator still depends on the chosen trajectory of the controlled evolution. This shortcoming serves as the primary motivation for the theoretical advances in the remainder of this work.

\subsection{Schedule-independent lower bounds on the annealing time}
\label{subsec:anneal}
To address the practical limitations of Eq.~\eqref{eq:genineq}, we now turn to the specific instance of quantum annealing in an open quantum system. We consider a time-dependent Hamiltonian of the form 
\begin{equation}
H(t)=H_{0}+f(t)H_{1},
\end{equation}
where we refer to $H_{0}$ as the drift Hamiltonian, $f(t)$ as the control field that describes the annealing schedule and $H_{1}$ as the corresponding control Hamiltonian. We start in the ground state of $H_1$ and seek to prepare the ground state of $H_0$. An adiabatic evolution would prepare the target state with high probability when $f(0) \gg 1$ such that the ground state of $H_0 + f(0) H_1$ is approximately the ground state of $H_1$.

Inspired by the approach in Ref.~\cite{Garc_a_Pintos_2025}, we can remove the explicit time dependence of the denominator in Eq.~\eqref{eq:genineq} by picking $\tilde{\mathcal L}_{t}\rho=-if(t)[H_{1},\rho]$. This choice gives a lower bound, 
\begin{equation}
    \label{eq:specineq}
    T \geq \frac{||\rho_T - \rho_0||_1}{ ||\Lind||_{1\rightarrow1}},
\end{equation}
 on the time $T$ it takes for annealing to prepare the state $\rho_{T}$ where
\begin{align}
 \Lind \rho = -i[H_0, \rho] + \mathcal{D}(\rho).
 \end{align}
We thus have derived a speed limit for quantum annealing in an open quantum system described by the master equation in Eq.~\eqref{eq:lindblad} that does not depend on the annealing schedule, an example of a trajectory-independent speed limit for nonunitary dynamics.

We remark that the induced 1-norm may be challenging to compute for superoperators. However, since $ ||\Lind||_{1\rightarrow1}\leq \sqrt{d}||\Lind||_{2\rightarrow2}$ where $d$ is the dimension of the system (see App.~\ref{app:normequiv}), we find another speed limit,  
\begin{equation}
    \label{eq:specineqweak}
    T \geq \frac{||\rho_T - \rho_0||_1}{ \sqrt{d}||\Lind||_{2\rightarrow2}},
\end{equation}
expressed in terms of the induced 2-norm that can be computed more readily (see App.~\ref{app:norms}). In App.~\ref{app:qslcomp}, we compare Eq.~\eqref{eq:specineq} when $\mathcal{D}=0$ to the schedule-independent bound for annealing in a closed quantum system derived in Ref.~\cite{Garc_a_Pintos_2025}.

\section{\label{sec:apps} Case studies}\label{sec:casestudies}
In this section we explore the tightness of the developed lower bound from Eq.~\eqref{eq:specineqweak} in examples. We compare the bound to the minimum time found numerically. To find this minimum time, we optimize a time-dependent control field $f(t)$--which functions like an annealing schedule--across various total times $T$, using a binary search to identify the shortest $T$. In order to perform this optimization, we assume the control field is piecewise constant on $20$ time intervals of equal length $\Delta t=T/20$, allowing us to easily optimize the control function using \texttt{scipy.optimize}'s BFGS \cite{2020SciPy-NMeth} for each total time $T$. We use the trace distance between the target state and the time-evolved state as the cost function for the BFGS optimizer with a target trace distance of 0.1. We note that there is no guarantee that the BFGS algorithm converges to a global optimum. Hence, the numerical estimate provided for the shortest total time $T$ can be considered as an upper bound on the minimum time required for comparison to the lower bound provided by Eq.~\eqref{eq:specineqweak}.

\subsection{\label{subsec:ampdamp} A single qubit subject to amplitude damping}
For systems of size $N$ with $N\geq 2$, the denominator of the bound in Eq.~\eqref{eq:specineqweak} in general cannot be written in closed form. Bounds for such systems will have to be treated either asymptotically or numerically. However, for $N=1$, an analytic expression for the bound can be computed exactly, providing a meaningful starting point to better understand the derived speed limit. We thus first consider a single qubit described by the Hamiltonian 
\begin{align}\label{eq:singlequbitham}
H(t)=-\omega \sigma_{z}+f(t)\sigma_{x}
\end{align}
subject to amplitude damping described by the dissipator 
\begin{equation}
\label{eq:ampdamp}
    \mathcal{D}(\rho) = \gamma \left( 2 \sigma_- \rho \sigma_+ - \sigma_+ \sigma_- \rho - \rho \sigma_+ \sigma_- \right).
\end{equation}
Here, $\sigma_{x}$ and $\sigma_{z}$ are the Pauli spin operators, $\sigma_- = |0\rangle \langle 1|$ and $\sigma_+ = \sigma_-^\dagger$ are the lowering and raising operators, and $\omega>0$ and $\gamma>0$ are the frequency and dissipation rate of the qubit, respectively. We assume that the qubit is initialized in the ground state $\rho_{0}=|-\rangle\langle -|$ of the control Hamiltonian $H_{1}=\sigma_{x}$. The target state $\rho_{T}=|0\rangle\langle 0|$ is the unique fixed point of the dissipator, i.e., $\mathcal D(\rho_{T})=0$. 

Since $||\rho_T - \rho_0||_1 = 1$, the bound in Eq.~\eqref{eq:specineqweak} then gives the lower bound
\begin{align}
    \label{eq:twolevelineq}
    T \geq \frac{1}{ \sqrt{2}|\gamma+2i\omega| }
\end{align}
for the time $T$ to prepare the target state $\rho_{T}$.  We note that this bound captures one trivial notion of reachability. As $\gamma$ and $\omega$ go to zero, the time required to prepare $\rho_{T}$ goes to infinity. Without the presence of the drift  Hamiltonian $H_{0}=-\omega\sigma_{z}$ and the dissipator $\mathcal D$, the control Hamiltonian $H_{1}=\sigma_{x}$ alone cannot drive the system to the target state.

\begin{figure}[t]
    \centering
    \includegraphics[width=\linewidth]{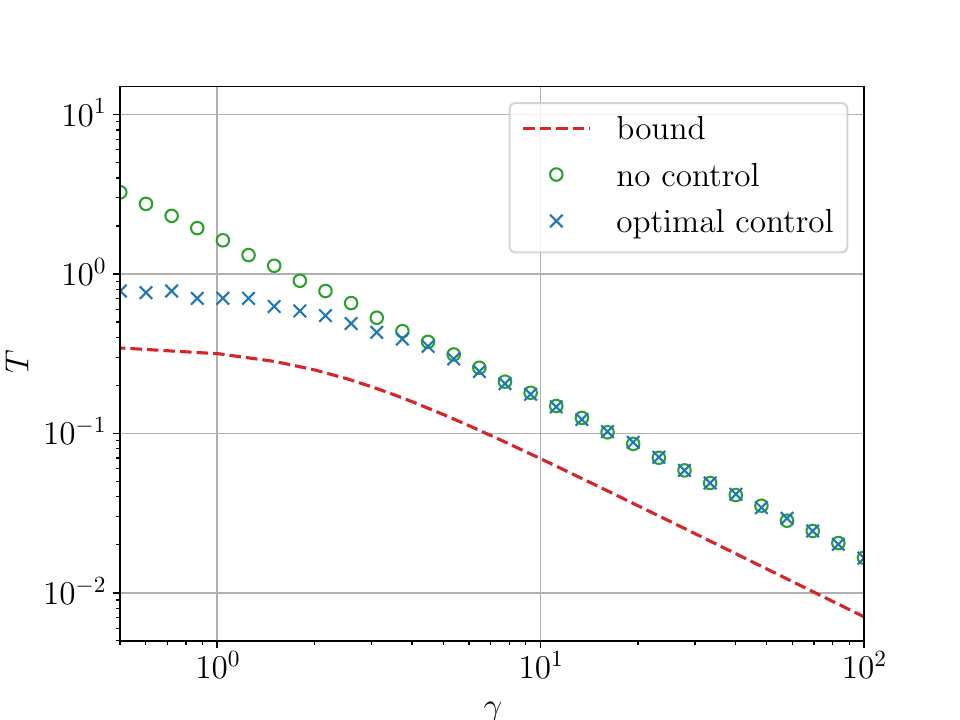}
    \caption{Time $T$ to prepare a desired target state for the single-qubit system described by Eqs. \eqref{eq:singlequbitham} and \eqref{eq:ampdamp} as a function of the decay rate $\gamma$ for fixed $\omega=1$. The bound in Eq.~\eqref{eq:twolevelineq} is shown as a dashed red line.  The blue markers show the time required to prepare the target state up to a trace distance of $0.1$ with a piecewise-constant control $f(t)$. The optimal time is numerically estimated using the method described in the introduction of Sec.~\ref{sec:casestudies}. For reference, the green markers show the time required for the dissipator alone to drive the system to its fixed point up to the same trace distance as in the case assisted by coherent control. } 
    \label{fig:gammascaling_om1}
\end{figure}

The scaling of this bound with respect to $\gamma$ is illustrated in Fig.~\ref{fig:gammascaling_om1} for $\omega = 1$, where the blue markers provide comparison to the numerically found minimum time required to prepare the target state. The plots indicate two separate regimes depending on the size of $\gamma$ relative to $\omega$. For $\gamma \gg \omega$, the strength of the dissipator dominates both the bound and the actual time required, hindering the acceleration possible with the control field $f(t)$. In the smaller $\gamma$ regime, the control is able to more dramatically accelerate the preparation of the target state.

\begin{figure}
    \centering
    \includegraphics[width=1\linewidth]{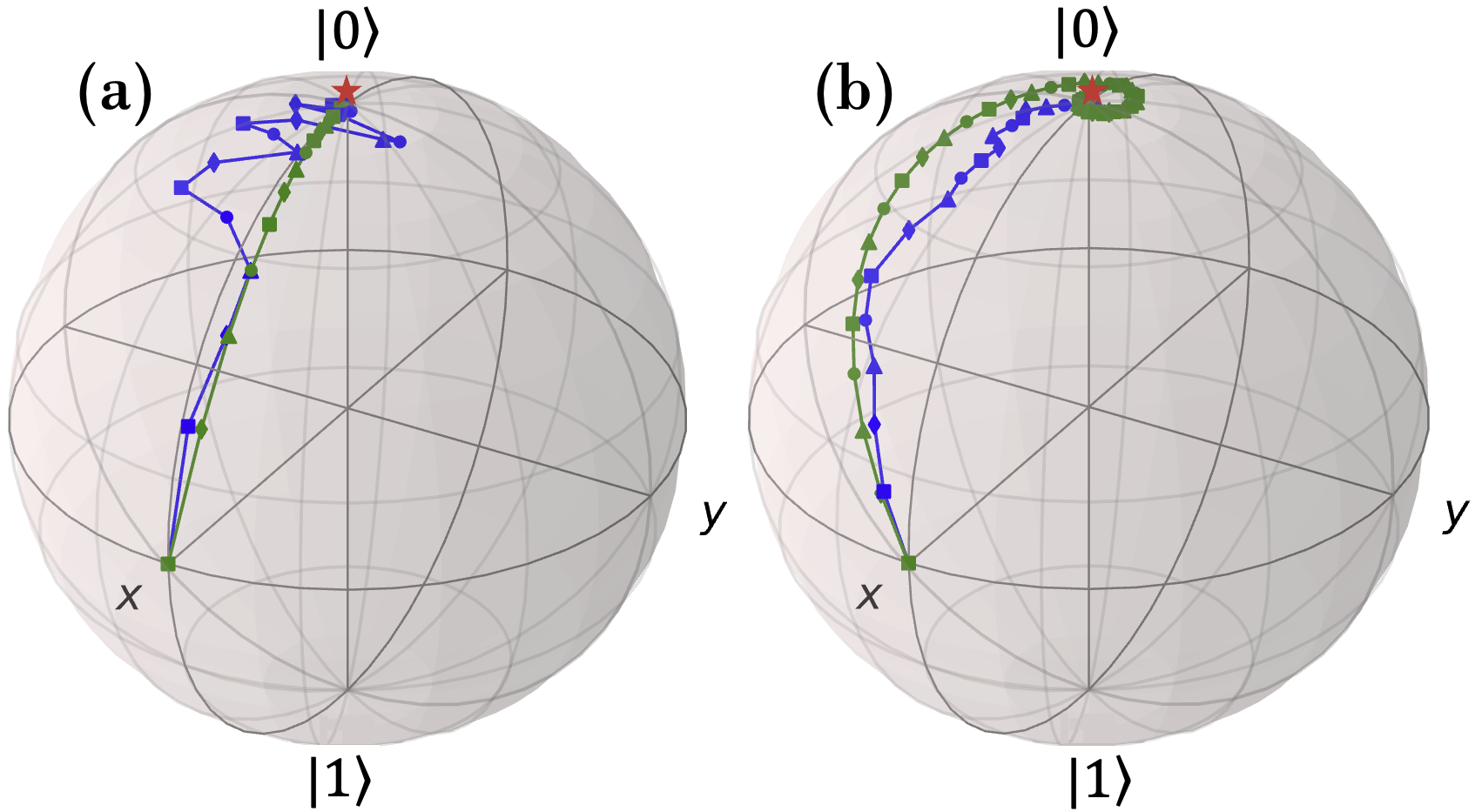}
    \caption{Controlled (blue) and uncontrolled (green) trajectories subject to amplitude damping with $\gamma = 1$ for the single-qubit system described by Eq.~\eqref{eq:singlequbitham} and \eqref{eq:ampdamp}. In (a), no drift Hamiltonian is present, i.e., $\omega=0$, while in (b) we set $\omega=1$. 
Even though the green trajectory in (b) spirals around the target state due to the presence of the drift Hamiltonian, the total preparation time is still equal to the preparation time that correspond to the trajectory in (a), highlighting that the presence of the drift Hamiltonian does not change the required preparation time. The blue trajectory in (b) illustrates that drift, control and amplitude damping are able to work together to accelerate the preparation of the target state 2.7-fold.   The green and blue markers are synchronized in time to help illustrate that the control is unable to accelerate the state preparation without the drift Hamiltonian. The red stars on the north poles illustrate the target states. These visualizations and the one-qubit dynamics are generated using QuTiP \cite{qutip5}.}
    \label{fig:bloch_traj}
\end{figure}

In Fig. \ref{fig:bloch_traj} we show an example of a controlled (blue) and uncontrolled (green) trajectory for (a) $\omega=0$ and (b) $\omega=1$ with fixed $\gamma =1$. The trajectories suggest that in the driftless case, i.e., when $\omega=0$, the control is unable to accelerate the time to prepare the target state, shown as a red marker.   
That is, the total time $T$ required is independent of the control field. Note that both green trajectories take the same amount of time even though the exact trajectory varies. This phenomenon generalizes to other values of $\omega$, highlighting the fact that the strength of the drift Hamiltonian has no effect on the overall time required for amplitude damping to prepare the target state in the absence of coherent control. We further note from Fig.~\ref{fig:bloch_traj}(b) that the control and drift combined are able to accelerate the target state preparation time, providing another illustration of the gap between the blue and green markers in Fig.~\ref{fig:gammascaling_om1}.

\subsection{\label{subsec:bellstate} Dissipative Bell state preparation} 
Next we consider an interacting two qubit system described by the Hamiltonian 
\begin{align}
\label{eq:hamtwoqubits}
H(t)=\omega(\sigma_1^{x}\sigma_{2}^{x}+\sigma_1^{z}\sigma_{2}^{z})+f(t)(\sigma_{1}^{x}+\sigma_{2}^{x}),
\end{align}
where $\sigma_{j}^{x}$ and $\sigma_{j}^{z}$ are Pauli operators acting on the $j$-th qubit. We assume that the system is initially prepared in the ground state $\rho_{0}=|-\rangle\langle -|\otimes |-\rangle\langle -| $ of the control Hamiltonian $H_{1}=\sigma_{1}^{x}+\sigma_{2}^{x}$. We aim to prepare a Bell state $\rho_{T}=|\beta_{3}\rangle \langle \beta_{3}|$ with the help of dissipation described by 
\begin{equation}
\label{eq:dissipatortwoqubits}
    \mathcal{D}(\rho) = \gamma\sum_{k=0}^2 \left( 2 \sigma_-^{(k)} \rho \sigma_+^{(k)} - \sigma_+^{(k)} \sigma_-^{(k)} \rho - \rho \sigma_+^{(k)} \sigma_-^{(k)} \right),
\end{equation}
where the jump operators $\sigma_-^{(k)} = |00\rangle\langle\beta_k|$ and $\sigma_{+}^{(k)}=\left(\sigma_{-}^{(k)}\right)^\dagger$  are formed by the four Bell states 
\begin{align*}
    |\beta_0 \rangle &= \frac{1}{\sqrt{2}} \left( |00\rangle + |11\rangle\right),\qquad |\beta_1 \rangle = \frac{1}{\sqrt{2}} \left( |01\rangle + |10\rangle\right), \\
    |\beta_2 \rangle &= \frac{1}{\sqrt{2}} \left( |00\rangle - |11\rangle\right), \qquad|\beta_3 \rangle = \frac{1}{\sqrt{2}} \left( |01\rangle - |10\rangle\right).
\end{align*}
We note that the target Bell state $|\beta_{3}\rangle$ is the unique fixed point of $\mathcal D$. The bound in Eq.~\eqref{eq:specineqweak} then provides a lower bound for the time $T$ to prepare $|\beta_{3}\rangle$.  

In Fig.~\ref{fig:gammascaling_bell} we analyze the tightness of the bound, which we compute numerically. The numerical results suggest that the bound captures the asymptotic scaling as $\gamma \to \infty$,  matching the intuitive prediction that stronger dissipation enables faster Bell state preparation. We note that the bound from Eq.~\eqref{eq:specineqweak} holds for any control Hamiltonians for which the initial state is an eigenstate. We explored this possibility by optimizing over three different independent control functions coupled to $\sigma_1^x$, $\sigma_2^x$ and $\sigma_1^x\sigma_2^x$. However, we were unable to find any speedup in the target Bell state preparation time with these extra degrees of control beyond what was already achieved with a single control field $f(t)$ coupled to the system via $H_{1}=\sigma_1^x+\sigma_2^x$.

\begin{figure}[t]
    \centering
    \includegraphics[width=\linewidth]{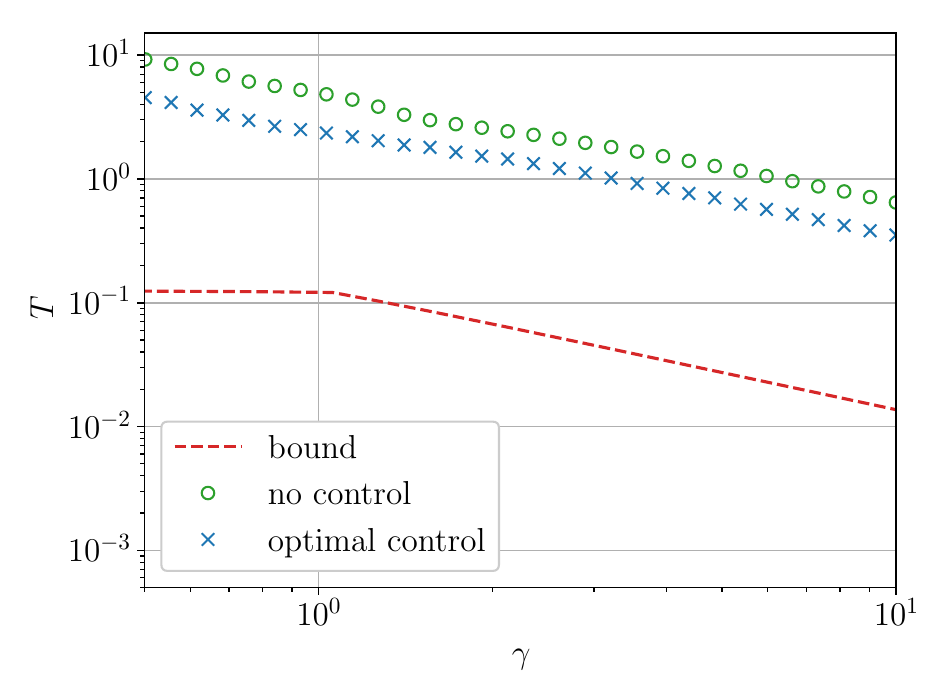}
    \caption{Time $T$ to prepare a Bell state for the two-qubit system described by Eqs.~\eqref{eq:hamtwoqubits} and \eqref{eq:dissipatortwoqubits} as a function of the dissipation rate $\gamma$ for fixed $\omega=1$. The bound in Eq.~\eqref{eq:specineqweak} is shown as a red dashed line. The blue markers show the time to prepare the target state up to trace distance 0.1 with a piecewise-constant control. The optimal time is numerically estimated from above. The green markers show the time required for the dissipator alone to drive the system to its fixed point up to the same trace distance as in the case assisted by coherent control.} 
    \label{fig:gammascaling_bell}
\end{figure}

\subsection{\label{subsec:thermal}Thermal state preparation}
Thermal state preparation is a fundamental subroutine for quantum simulation~\cite{Chen2025,Chen2025b}.
It has wide ranging applications in many-body physics, equilibrium statistical mechanics and materials science and can enable the study of phase transitions \cite{Alhambra2021LectureNotes}. For example, it can be leveraged to study the Mpemba effect \cite{Edo2025}. In this section, we will examine the behavior of the speed limit from Eq.~\eqref{eq:specineqweak} for this subroutine. 

We consider a system of $N$ spins described by the Ising Hamiltonian
\begin{equation}
\label{eq:hamising}
    H(t) = -\sum_{i=1}^N h_i \sigma_i^z + \sum_{i,j=1}^N J_{ij}\sigma_i^z\sigma_j^z - f(t)\sum_{i=1}^N \sigma_i^x,
\end{equation} 
where $h_i$ are the local fields on each spin, $J_{ij}$ are the coupling constants, and the control field $f(t)$ is given by the transverse magnetic field amplitude. We focus on the case of the extensive antiferromagnetic Ising model with all-to-all couplings (infinite range) with $h_i = 1$ and $J_{ij} = \frac{1}{N}$ for $i,j = 1,...,N$. The uniform antiferromagnetic coupling gives rise to frustration in the ground state, which is not resolved by the non-zero global field.

When the spins  with $f(t)=0$ are weakly coupled to a bath in thermal equilibrium at inverse-temperature $\beta$, the dynamics can be approximated under suitable assumptions by a master equation of the form of Eq.~\eqref{eq:lindblad} where the dissipator reads~\cite{Davies1974}

    \begin{equation}
    \label{eq:dissipatorising}
    \mathcal{D}(\rho) = \sum_{a,b,\omega}  \gamma_{ab}(\omega) \left( L_{\omega, b} \rho L_{\omega, a}^\dagger - \frac{1}{2}\left\{L_{\omega,a}^\dagger L_{\omega,b}, \rho\right\}\right).
\end{equation}
Here, the jump operators are given by  $L_{\omega, k} = \sum_{\varepsilon_b - \varepsilon_a = \omega} |\varepsilon_a\rangle \langle \varepsilon_a | \sigma_k^x |\varepsilon_b \rangle \langle \varepsilon_b|$ where $|\varepsilon_{a}\rangle$ are the eigenstates of the drift Hamiltonian $H_{0}=-\sum_{i=1}^N h_i \sigma_i^z + \sum_{i,j=1}^N J_{ij}\sigma_i^z\sigma_j^z$ while the decay rates $\gamma_{ab}$'s are related to the spectral density functions. Specifically, in order for the thermal state $\rho \propto e^{-\beta H_0}$ to be a fixed point of $\mathcal D$, $\gamma_{a b}(\omega)$ must satisfy the Kubo-Martin-Schwinger detailed balance condition~\cite{Kubo1957,Martin1959,Haag1967} $\gamma_{a b}(-\omega) = \gamma_{b a}(\omega) e^{-\beta \omega}$. We can satisfy this condition by choosing 
\begin{equation} \label{eqt:gamma}
    \gamma_{ab}(\omega) = \frac{2\pi \omega e^{-|\omega|/\omega_c}}{1-e^{-\beta \omega}} \eta g^2 \delta_{ab} ,
\end{equation}
where $\omega_c$ is a high-frequency cut-off, $\eta g^2$ is a constant corresponding to the coupling strength between the bath and the qubit.
This form arises from a specific model for the thermal bath and system-bath coupling: each qubit interacts with an independent harmonic oscillator bath, and each bath has an identical Ohmic spectral density with a high frequency cut-off.
 A derivation of Eq.~\eqref{eqt:gamma} for this bath model can be found in Ref.~\cite{Albash_2012}.

\begin{figure}
    \centering
    \includegraphics[width=\linewidth]{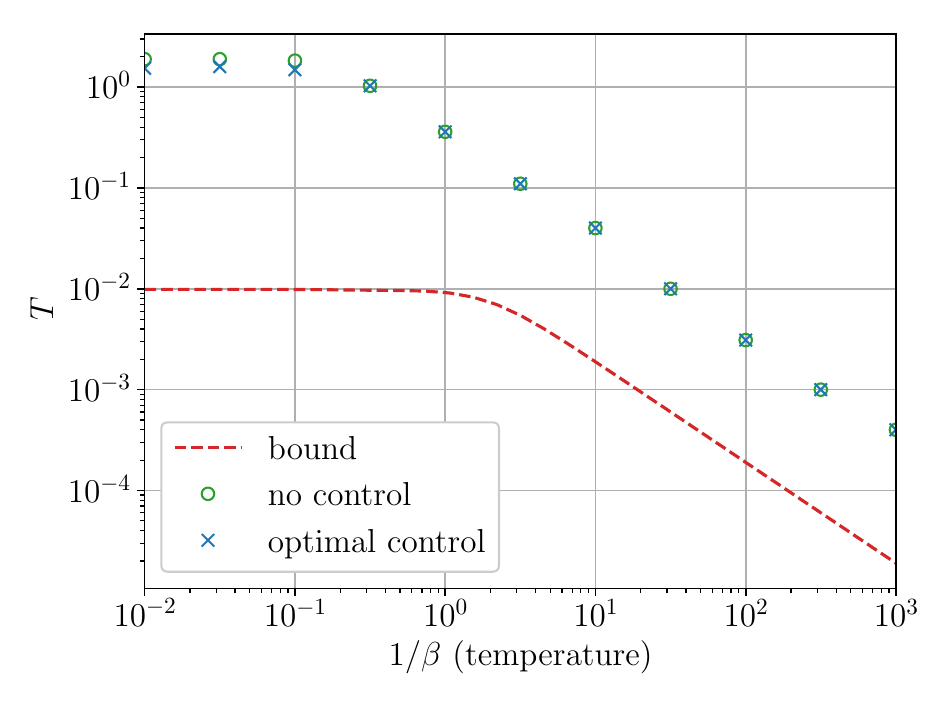}
    \caption{Time $T$ to prepare the thermal state for the four-qubit system described by Eqs.~\eqref{eq:hamising} and \eqref{eq:dissipatorising}--i.e. an extensive all-to-all antiferromagnetic Ising model coupled to a bath of harmonic oscillators in thermal equilibrium--as a function of the temperature $1/\beta$. The bound in Eq.~\eqref{eq:specineqweak} is shown as a red dashed line. The blue markers show the time to prepare the target state up to trace distance 0.1 with a piecewise-constant control. The optimal time is numerically estimated from above. For reference, the green markers show the time required for the dissipator alone to drive the system to its fixed point up to the same trace distance as in the case assisted by coherent control.}
    \label{fig:thermal_bound_4}
\end{figure}

In Fig. \ref{fig:thermal_bound_4} we numerically explore the bound in Eq.~\eqref{eq:specineqweak} for thermal state preparation as a function of the temperature $1/\beta$. We considered the Ising model (Eq.~\eqref{eq:hamising}) and dissipator given in Eq.~\eqref{eq:dissipatorising} for an $N=4$ spin system initially prepared in the state $\rho_{-}=|-\rangle\langle -|^{\otimes N}$.

We observe that the transverse-field control only noticeably accelerates state preparation in the small temperature regime. In all other cases we could not find optimal controls that lead to faster thermal state preparation times $T$.  In the high temperature limit (i.e. $\beta \to 0$), we have that $\gamma_{ab}(\omega) \sim \frac{1}{\beta}$ (i.e. the spectral density function is proportional to temperature). No other terms in the bound given by Eq.~\eqref{eq:specineq} change with temperature, so $||\Lind||_{2\rightarrow2}$
grows proportionally to the temperature. As a result, the bound goes to zero as the temperature goes to infinity. We find similar asymptotic behavior for the numerically determined state preparation time, suggesting that this aspect of the bound captures that systems will thermalize much faster at higher temperatures.

\section{\label{sec:conclusion}Conclusion}
In this work, we have derived a trajectory-independent speed limit for annealing in open quantum systems, directly generalizing results from Ref.~\cite{Garc_a_Pintos_2025}. The derived lower bounds on the annealing time (Eqs.~\eqref{eq:specineq} and \eqref{eq:specineqweak}) describe the extent to which the drift Hamiltonian and dissipator (i.e., the uncontrolled dynamics) in a Lindblad master equation are needed to reach a desired target state with the help of coherent control. As such, the developed bounds resemble speed limits obtained for closed quantum systems that draw on relations to the uncontrolled evolution \cite{Arenz_2017,Lee_2018, PhysRevA.105.042402, Wiedmann2026}.    

We have demonstrated analytically and numerically that the lower bounds we derived capture key asymptotic scaling with respect to the strength of the dissipator. In the case of thermal state preparation, this scaling implies that the lower bound properly accounts for the expected behavior in the high temperature limit.

We envision a number of future directions of inquiry. Most notably, the lower bound in Eq.~\eqref{eq:specineq} decays with the system size, counter to the common observation that the annealing time typically increases with the system size. This behavior mirrors the shortcoming discussed in Ref.~\cite{Garc_a_Pintos_2025} for the closed-system case, and it remains an open problem whether trajectory-independent bounds can properly capture system-size-dependent scaling behavior.  The bounds derived herein could also potentially be tightened by leveraging different distance metrics between the initial and final states, as in Refs.~\cite{Pires2016PRX} and \cite{Campaioli2019tightrobust}, or by further studying the CPTP dynamics in Liouvillian space as in Ref.~\cite{Srivastav_2025}. Furthermore, given the generic form of Eq.~\eqref{eq:genineq}, it would be interesting to explore applications of such speed limits beyond the annealing framework. For example, given the recent advances in quantum algorithms based on dissipative state preparation protocols \cite{Lin2025}, an asymptotic match between the time complexity of a dissipative quantum state preparation algorithm and a speed limit could provide statements of optimality.\\

\begin{acknowledgments}
J.L. would like to thank Eric Bobrow, Tianyue Li, and Huston Wilhite for helpful discussions. J.L. acknowledges support from the Department of Energy
Computational Science Graduate Fellowship under Award Number DE-SC0024386. C.A. acknowledges support from the National Science Foundation (Grant No. 2231328).

This work was supported by the Laboratory Directed Research
and Development program (Project 233972)
at Sandia National Laboratories, a multimission laboratory managed and operated by National Technology and
Engineering Solutions of Sandia LLC, a wholly owned
subsidiary of Honeywell International Inc. for the U.S.
Department of Energy’s National Nuclear Security Administration under contract DE-NA0003525. This paper
describes objective technical results and analysis. Any
subjective views or opinions that might be expressed in
the paper do not necessarily represent the views of the
U.S. Department of Energy or the United States Government. SAND2026-18384O.
\end{acknowledgments}

\bibliography{bib}

\newpage
\onecolumngrid
\appendix

\section{\label{app:norms}Norm Notation}

Throughout this work, for operators we use the trace norm (Schatten 1-norm) $||A||_1 = \text{Tr}(\sqrt{A^\dagger A})$ and the Hilbert-Schmidt norm (Schatten 2-norm) $||A||_2 = \sqrt{\text{Tr}(A^\dagger A)}$. For superoperators, we adopt the notation of Ref.~\cite{Watrous2005} for the following induced norms:
\begin{align}
    ||\mathcal{B}||_{1\rightarrow1} &= \sup_{A \neq 0} \frac{||\mathcal{B}A||_1}{||A||_1} = \sup_{||A||_1 = 1} \text{Tr}\left(\sqrt{(\mathcal{B}A)^\dagger (\mathcal{B}A)}\right), \\
    ||\mathcal{B}||_{2\rightarrow2} &= \sup_{A \neq 0} \frac{||\mathcal{B}A||_2}{||A||_2} = \sup_{||A||_2 = 1} \sqrt{\text{Tr}\left((\mathcal{B}A)^\dagger (\mathcal{B}A)\right)}.
\end{align}
It is straightforward to compute $||\mathcal{B}||_{2\rightarrow2}$ by writing the superoperator as a matrix acting on a vectorized density matrix and taking the largest singular value of that matrix. We can follow the standard vectorization conventions outlined in Sec.~4.3 of Ref.~\cite{Horn1991} for this process, i.e.,
\begin{equation}
    \text{vec}(\rho) = \begin{pmatrix}
        \rho_{i1}\\ \vdots \\ \rho_{i2} \\ \vdots
    \end{pmatrix},\quad
    \text{vec}(ABC) = \left( C^T \otimes A \right) \text{vec}(B).
\end{equation}

\section{\label{app:normequiv}Norm Equivalence}
In this Appendix, we prove Eq.~\eqref{eq:specineqweak} by showing that $||\mathcal{B}||_{1\rightarrow1} \leq \sqrt{2^N}||\mathcal{B}||_{2\rightarrow2}$. This enables us to use the computationally easier induced 2-norm for the numerical simulations. The result below follows standard methods from Chapter~3 of Ref.~\cite{Humpherys2017} for matrix norm equivalences, adapted to the superoperator case.
\begin{prop}
\label{prop:normequiv}
    If $\mathcal{B}$ is a linear superoperator acting on a $d^2$-dimensional vector space of square matrices, then
    \begin{equation}
    \label{eq:normequiv}
        ||\mathcal{B}||_{1\rightarrow1} \leq\sqrt{d}||\mathcal{B}||_{2\rightarrow2}.
    \end{equation}
\end{prop}
\begin{proof}
    Let $\{\sigma_k\}_{k=1}^d$ denote the singular values of $A$. First, note that 
    \begin{equation}
        ||A||_2^2 = \sum_{k=1}^d\sigma_k^2 \leq \left(\sum_{k=1}^d \sigma_k\right)^2 = ||A||_1^2,
    \end{equation}
    so $||A||_2 \leq ||A||_1$. Similarly, by Cauchy-Schwarz,
    \begin{equation}
        ||A||_1 = \sum_{k=1}^d \sigma_k\cdot 1 \leq \sqrt{\sum_{k=1}^d 1}\sqrt{\sum_{k=1}^d\sigma_k^2} = \sqrt{d}||A||_2.
    \end{equation}
    Therefore, we have that
    \begin{equation}
        ||\mathcal{B}||_{1\rightarrow1} = \sup_{A\neq 0} \frac{||\mathcal{B}A||_1}{||A||_1}
        \leq \sup_{A\neq 0} \frac{\sqrt{d}||\mathcal{B}A||_2}{||A||_1}
        \leq \sup_{A\neq 0} \frac{\sqrt{d}||\mathcal{B}A||_2}{||A||_2} = \sqrt{d}||\mathcal{B}||_{2\rightarrow2}.
    \end{equation}
\end{proof}

\section{\label{app:qslcomp}Speed Limit Comparison in Closed-System Case}

We can compare Eq.~\eqref{eq:specineq} to the bound from \cite{Garc_a_Pintos_2025} for the case when $\mathcal{D} = 0$ and the initial and target states are pure states, i.e. $\rho_0=|\psi_0 \rangle \langle \psi_0|$ and $\rho_T = |\psi_T \rangle \langle \psi_T|$. With this setup, Eq.~(5) from \cite{Garc_a_Pintos_2025} takes the following form:
\begin{equation}
    T \geq \frac{D_B\left( |\psi_T\rangle, |\psi_0\rangle\right)}{\sqrt{\text{Var}_{|\psi_T\rangle}\left(H_0 \right)}},
\end{equation}
where $D_B(|\psi_T\rangle,|\psi_0\rangle) = \sqrt{2\left(1-|\langle\psi_T|\psi_0\rangle|\right)}$ is the Bures distance between the initial and final state and $\text{Var}_{|\psi_T\rangle}\left(H_0 \right) = \langle \psi_T | H_0^2|\psi_T \rangle - \left( \langle \psi_T | H_0|\psi_T \rangle\right)^2$ is the variance of $H_0$ with respect to the final state. For pure states, we have that $||\rho_T - \rho_0||_1 = 2\sqrt{1-|\langle\psi_T|\psi_0\rangle|^2}$, giving us the following equivalence:
\begin{equation}
    ||\rho_T - \rho_0||_1 \leq 2D_B\left( |\psi_T\rangle, |\psi_0\rangle\right) \leq \sqrt{2}||\rho_T-\rho_0||_1.
\end{equation}
Furthermore, when $\mathcal{D} = 0$ we can compute that $||\Lind||_{1\rightarrow1} = E_{\text{max}}(H_0) - E_{\text{min}}(H_0) = ||\Lind||_{2\rightarrow2}$, i.e. the difference between the maximum and minimum eigenvalues of the drift Hamiltonian. This case of unitary dynamics loosens the inequality from Eq.~\eqref{eq:normequiv}. Popoviciu's inequality \cite{popoviciu1935sur} applied to the variance of a random variable states that 
\begin{equation}
\sqrt{\text{Var}_{|\psi_T\rangle}\left(H_0 \right)} \leq \frac{E_{\text{max}}(H_0) - E_{\text{min}}(H_0)}{2}.
\end{equation}
Combining these inequalities, we get
\begin{equation}
\label{eq:usvsluis}
    \frac{||\rho_T - \rho_0||_1}{ ||\Lind||_{1\rightarrow1}} \leq  \frac{D_B\left( |\psi_T\rangle, |\psi_0\rangle\right)}{\sqrt{\text{Var}_{|\psi_T\rangle}\left(H_0 \right)}}.
\end{equation}
Hence, the generalized bound in Eq.~\eqref{eq:specineq} is looser than \cite{Garc_a_Pintos_2025} when restricted to the closed quantum system case due to using the trace distance to generalize to mixed states.

\end{document}